\documentclass[a4paper]{article}

\usepackage{graphicx}

\usepackage{amsmath,amssymb}
\usepackage{url}
\usepackage{rotating}
\usepackage{comment}
\usepackage{multirow}
\usepackage{amsthm}
\theoremstyle{plain}
\newtheorem{theorem}{Theorem}

\newtheorem{lemma}[theorem]{Lemma}

\newcommand{\emptystr}{\varepsilon}
\newcommand{\SA}{\mathsf{SA}}
\newcommand{\LF}{\mathsf{LF}}
\newcommand{\Lstr}{\mathsf{L}}
\newcommand{\Fstr}{\mathsf{F}}
\newcommand{\palSA}{\mathsf{SA}_{\mathsf{pal}}}
\newcommand{\palLF}{\mathsf{LF}_{\mathsf{pal}}}
\newcommand{\palLstr}{\mathsf{L}_{\mathsf{pal}}}
\newcommand{\palFstr}{\mathsf{F}_{\mathsf{pal}}}

\newcommand{\rank}{\mathsf{rank}}
\newcommand{\select}{\mathsf{select}}
\newcommand{\rangecount}{\mathsf{rangeCount}}

\newcommand{\occ}{\mathsf{occ}}
\newcommand{\ssp}{\mathsf{ssp}}
\newcommand{\sspg}{\mathsf{sspg}}
\newcommand{\nsspg}{\mathsf{G}}
\newcommand{\spp}{\mathsf{spp}}
\newcommand{\sppg}{\pi}

\newcommand{\entropy}{\mathcal{H}}
\newcommand{\rMq}{\mathsf{RMQ}}

\newcommand{\rev}[1]{#1^R}

\newcommand{\lpal}{\mathsf{lpal}}

\newcommand{\idtt}[1]{\ensuremath{\mathtt{#1}}}

\title{PalFM-index: FM-index for Palindrome Pattern Matching} 

\author{
  Shinya Nagashita\\
  {Kyushu Institute of Technology, Japan}\\
  {\texttt{nagashita.shinya206@mail.kyutech.jp}}\\
  \\
  Tomohiro~I\\
  {Kyushu Institute of Technology, Japan}\\
  {\texttt{tomohiro@ai.kyutech.ac.jp}}\\
}

\date{}
\begin{document}

\maketitle

\begin{abstract}
The palindrome pattern matching (pal-matching) is a kind of generalized pattern matching,
in which two strings $x$ and $y$ of same length are considered to match (pal-match) if they have the same palindromic structures, i.e.,
for any possible $1 \le i < j \le |x| = |y|$, $x[i..j]$ is a palindrome if and only if $y[i..j]$ is a palindrome.
The pal-matching problem is the problem of searching for, in a text, the occurrences of the substrings that pal-match with a pattern.
Given a text $T$ of length $n$ over an alphabet of size $\sigma$, 
an index for pal-matching is to support, given a pattern $P$ of length $m$,
the counting queries that compute the number $\occ$ of occurrences of $P$ and
the locating queries that compute the occurrences of $P$.
The authors in~[I et al., Theor. Comput. Sci., 2013] proposed an $O(n \lg n)$-bit data structure
to support the counting queries in $O(m \lg \sigma)$ time and the locating queries in $O(m \lg \sigma + \occ)$ time.
In this paper, we propose an FM-index type index for the pal-matching problem, which we call the PalFM-index,
that occupies $2n \lg \min(\sigma, \lg n) + 2n + o(n)$ bits of space and
supports the counting queries in $O(m)$ time.
The PalFM-indexes can support the locating queries in $O(m + \Delta \occ)$ time
by adding $\frac{n}{\Delta} \lg n + n + o(n)$ bits of space,
where $\Delta$ is a parameter chosen from $\{1, 2, \dots, n\}$ in the preprocessing phase.
\end{abstract}

\section{Introduction}\label{sec:intro}
A palindrome is a string that can be read same backward as forward.
Palindromic structures in a string are one of the most fundamental structures in the string and have been extensively studied. 
For example, it is known that any string $w$ contains at most $|w| + 1$ distinct palindromic substrings~\cite{2001DroubayJP_EpistWordsAndSomeConst},
and the strings reaching the maximum values have some intriguing properties~\cite{2009GlenJWZ_PalinRichn,2009RestivoR_BurrowWheelTransAndPalin}.
Another concept regarding palindromic structures is the palindrome complexity~\cite{2003AlloucheBCD_PalinCompl,2004BrlekHNR_PalinComplOfInfinWords,2010AnisiuAK_TotalPalinComplOfFinit},
which is the number of distinct palindromic substrings of a given length in a string.

Instead of thinking about distinct palindromic substrings,
one might be interested in occurrences of palindromic substrings.
The palindromic structures in such a sense are captured by the maximal palindromes from all possible ``centers'' in a string.
Manacher's algorithm~\cite{1975Manacher_NewLinearTimeOnLine}, originally proposed for computing a prefix-palindrome,
can be extended to compute all the maximal palindromes in $O(|w|)$ time for a string $w$.
The authors in~\cite{2010IIBT_CountAndVerifMaximPalin_SPIRE} considered 
the problem of inferring strings from a given set of maximal palindromes
and showed that the problem can be solved in $O(|w|)$ time.

In~\cite{I2013Ppm}, a new concept called \emph{palindrome pattern matching} was introduced as a generalized pattern matching.
Two strings $x$ and $y$ of the same length are said to \emph{palindrome pattern match} (\emph{pal-match} in short) 
iff they have the same palindromic structures, i.e., the following condition holds:
for any possible $1 \le i < j \le |x| = |y|$, $x[i..j]$ is a palindrome iff $y[i..j]$ is a palindrome.
We remark that $x$ and $y$ themselves are not necessarily palindromes.
The palindrome pattern matching has potential applications to genomic analysis,
in which some palindromic structures play an important role to estimate RNA secondary structures~\cite{1971TinocoUL_EstimOfSeconStrucIn}.

The pal-matching problem is to search for, in a text, the occurrences of the substrings that pal-match with a pattern.
Given a text $T$ of length $n$ and a pattern $P$ of length $m$, 
a Morris-Pratt type algorithm for solving the pal-matching problem in $O(n)$ time was proposed in~\cite{I2013Ppm}.
The method in~\cite{I2013Ppm} is based on the $\lpal$-encoding of a string $w$, denoted as $\lpal_{w}$,
that is the integer array of length $|w|$ such that $\lpal_{w}[i]$ is the length of the longest suffix palindrome of $w[1..i]$.
The $\lpal$-encoding is helpful because two strings $x$ and $y$ pal-match iff $\lpal_{x} = \lpal_{y}$.
When $T$ is large and static, and patterns come online later, one might think of preprocessing $T$ to construct an index for pal-matching.
An index for pal-matching is to support
the counting queries that compute the number $\occ$ of occurrences of $P$ and
the locating queries that compute the occurrences of $P$.
For this purpose, I et al.~\cite{I2013Ppm} proposed the \emph{palindrome suffix tree} of $T$, 
which is a compacted tree of the $\lpal$-encoded suffixes of $T$.
The palindrome suffix tree takes $O(n \lg n)$ bits of space and
supports the counting queries in $O(m \lg \sigma)$ time and the locating queries in $O(m \lg \sigma + \occ)$ time, 
where $\sigma$ is the size of the alphabet from which characters in $T$ are taken
and $\occ$ is the number of occurrences.

In this paper, we present a new index, named the \emph{PalFM-index}, by applying the technique of the FM-index~\cite{Ferragina2000ODS} to the pal-matching problem.
In so doing we introduce a new encoding, named the $\ssp$-encoding, that is based on the non-trivial shortest suffix-palindrome of each prefix.
In contrast to the $\lpal$-encoding, the $\ssp$-encoding has a good property to design the PalFM-index.
The PalFM-index occupies $2n \lg \min(\sigma, \lg n) + 2n + o(n)$ bits of space and supports the counting queries in $O(m)$ time.
The locating queries can be supported in $O(m + \Delta \occ)$ time
by adding $\frac{n}{\Delta} \lg n + n + o(n)$ bits of space,
where $\Delta$ is a parameter chosen from $\{1, 2, \dots, n\}$ in the preprocessing phase.

\subsection{Related work}
One of the well-studied algorithmic problems related to palindromes is factorizing a string into non-empty palindromes,
or in other words, recognizing a string that is obtained by concatenating a certain number of non-empty palindromes~\cite{1975Manacher_NewLinearTimeOnLine,KMP77,Galil1978LOR,2014FiciGKK_SubquadAlgorForMinimPalin_JDA,2014ISIBT_ComputPalinFactorAndPalin,2015KosolobovRS_PalKIsLinearRecog_SOFSEM,2017BorozdinKRS_PalinLengtInLinearTime_CPM,2018RubinchikS_EertrEfficDataStrucFor_EJC}.
The combinatorial properties discovered during tackling this factorization problem are useful to work on palindromes-related problems.

Developing techniques of designing space-efficient indexes for generalized pattern matching is of great interest.
Our PalFM-index was inspired by that of Kim and Cho~\cite{2021KimC_SimplFmIndexForParam},
which is a simplified version of the FM-index for parameterized pattern matching~\cite{2017GangulyST_PbwtAchievSuccinDataStruc_SODA}.
Indexes based on the FM-index for other generalized pattern matching problems were considered in~\cite{Ganguly2017StructuralPatternMatching,2017GagieMV_EncodForOrderPreserMatch_ESA,2021KimC_CompacIndexForCartesTree_CPM}.

\section{Preliminaries}\label{sec:prelim}

\subsection{Notations}
An integer interval $\{ i, i+1, \dots, j\}$ is denoted by $[i..j]$, 
where $[i..j]$ represents the empty interval if $i > j$.

Let $\Sigma$ be a finite \emph{alphabet}, a set of characters.
An element of $\Sigma^*$ is called a \emph{string}.
The length of a string $w$ is denoted by $|w|$. 
The empty string $\emptystr$ is a string of length 0,
that is, $|\emptystr| = 0$.
The concatenated string of two strings $x$ and $y$ are denoted as $x \cdot y$ or simply $xy$.
The $i$-th character of a string $w$ is denoted by $w[i]$ for $1 \leq i \leq |w|$,
and the \emph{substring} of a string $w$ that begins at position $i$ and
ends at position $j$ is denoted by $w[i..j]$ for $1 \leq i \leq j \leq |w|$,
i.e. $w[i..j] = w[i]w[i+1] \dots w[j]$.
For convenience, let $w[i..j] = \emptystr$ if $i > j$.
A substring of the form $w[1..j]$ (resp. $w[i..|w|]$) is called a \emph{prefix} (resp.\ \emph{suffix}) of $w$ and denoted as $w[..j]$ (resp. $w[i..]$) in shorthand.
Note that $\emptystr$ is a substring/prefix/suffix of any string $w$.
A substring of $w$ is called \emph{proper} if it is not $w$ itself.
When needed we use parentheses to indicate positions in a concatenated string, for example,
$(xy)[i]$ refers to the $i$-th character of the string $xy$. 
Hence, $(xy)[i]$ should be distinguished from $xy[i]$, which can be interpreted as the concatenated string of $x$ and $y[i]$.

Let $\prec$ denote the total order over an alphabet we consider.
In particular, we will consider strings over a set consisting of integers and $\infty$, in which natural total order based on their values is employed.
We extend $\prec$ to denote the lexicographic order of strings over the alphabet.
For any strings $x$ and $y$ that do not match,
we say that $x$ is lexicographically smaller than $y$ and denote it by $x \prec y$ iff 
$x[i+1] \prec y[i+1]$ for largest integer $i$ with $x[..i] = y[..i]$, 
where we assume that $x[i+1]$ or $y[i+1]$ refers to the lexicographically smallest character $\$$ if it points to out of bounds.

For any string $w$, let $\rev{w}$ denote the reversed string of $w$,
that is, $\rev{w} = w[|w|] \cdots w[2]w[1]$.
A string $w$ is called a \emph{palindrome} if $w = \rev{w}$.
The \emph{radius} of a palindrome $w$ is $\frac{|w|}{2}$.
The \emph{center} of a palindromic substring $w[i..j]$ of a string $w$ is 
$\frac{i+j}{2}$.
A palindromic substring $w[i..j]$ is called the \emph{maximal palindrome} at the center $\frac{i+j}{2}$
if no other palindromes at the center $\frac{i+j}{2}$ have a larger radius than $w[i..j]$,
i.e., if $w[i-1] \neq w[j+1]$, $i = 1$, or $j = |w|$.

Two strings $x$ and $y$ of same length are said to \emph{palindrome pattern match} (\emph{pal-match} in short) 
iff they have the same palindromic structures, i.e., the following condition holds:
for any possible $1 \le i < j \le |x| = |y|$, $x[i..j]$ is a palindrome iff $y[i..j]$ is a palindrome.
For example, $\idtt{abcbaaca}$ and $\idtt{bcacbbdb}$ pal-match since their palindromic structures coincide (see Figure~\ref{fig:pal-match}).
Note that pal-matching induces a substring consistent equivalent relation~\cite{2016MatsuokaAIBT_GenerPatterMatchAndPeriod_TCS}, i.e., 
if $x$ and $y$ pal-match then $x[i..j]$ and $y[i..j]$ pal-match for any possible $1 \le i < j \le |x| = |y|$.

\begin{figure}[t]
 \center{
   \includegraphics[scale=0.45]{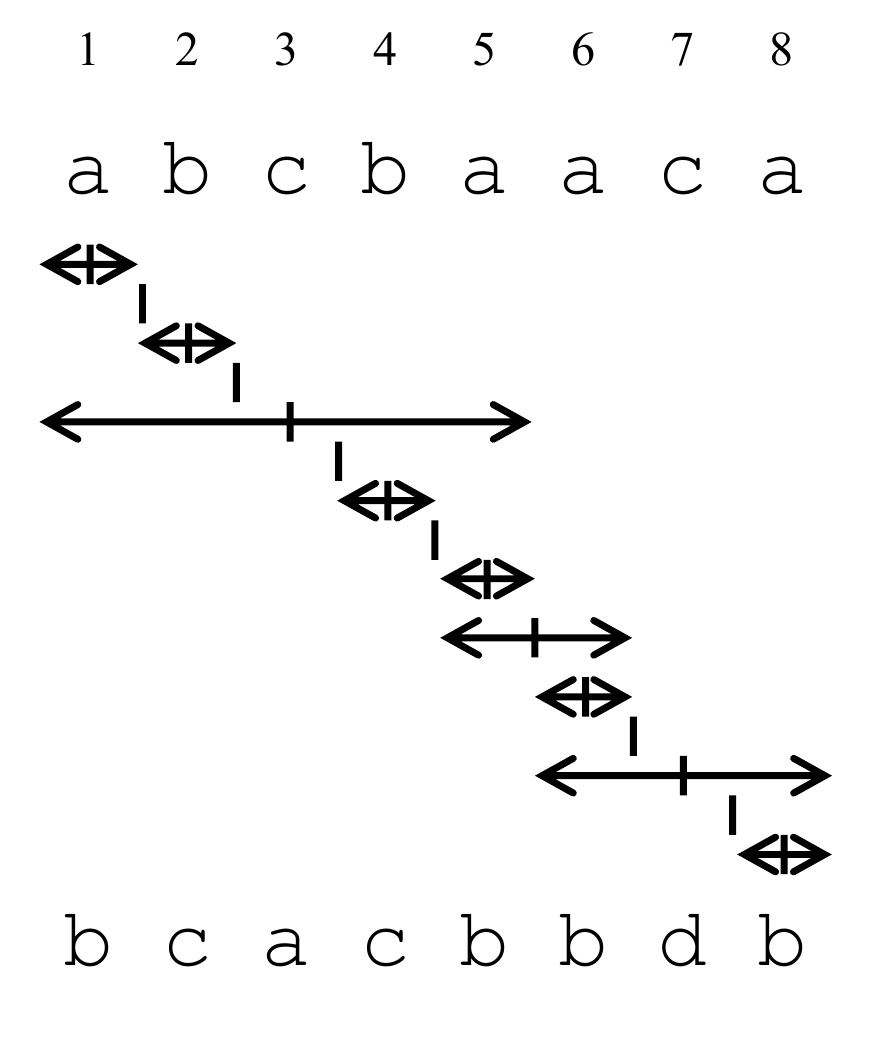}
 }
 \caption{Illustration of the palindromic structures for pal-matching strings $\idtt{abcbaaca}$ and $\idtt{bcacbbdb}$.
   Check that the radii of their maximal palindromes for all possible centers, which are illustrated by two-headed arrows, coincide.
 }
 \label{fig:pal-match}
\end{figure}

The pal-matching problem is to search for, in a text string $T$, the occurrences of the substrings that pal-match with a pattern $P$.
In the pal-matching problem, an occurrence of $P$ refers to a position $i$ such that $T[i..i+|P|-1]$ and $P$ pal-match.
Throughout this paper we consider indexing a text $T$ of length $n$ over an alphabet $\Sigma$ of size $\sigma$.

\subsection{Toolbox}
As a component of our PalFM-index, we use a data structure for a string $w$ 
over an integer alphabet $U$ supporting the following queries.
\begin{itemize}
  \item $\rank_w(i, c)$: return the number of occurrences of character $c \in U$ in $w[..i]$.
  \item $\select_w(i, c)$: return the $i$-th smallest position of the occurrences of character $c \in U$ in $w$.
  \item $\rangecount_w(i, j, c, d)$: return the number of the occurrences of any character in $[c..d] \subseteq U$ in $w[i..j]$.
\end{itemize}
The Wavelet tree~\cite{2003GrossiGV_HighOrderEntropComprText} supports these queries in $O(\lg |\Sigma|)$ time
using $|w| \entropy_0(w) + o(|w| \lg |U|)$ bits of space,
where $\entropy_0(w) = O(\lg |U|)$ is the 0-th order empirical entropy of $w$.
The subsequent studies~\cite{2007FerraginaMMN_ComprRepresOfSequenAnd,2008GolynskiRR_RedunOfSuccinDataStruc_SWAT}
improved the complexities, resulting in the following theorem.
\begin{theorem}[\cite{2008GolynskiRR_RedunOfSuccinDataStruc_SWAT}]\label{theo:wt}
  For a string $w$ over an integer alphabet $U$, there is a data structure in $|w| \entropy_0(w) + o(|w|)$ bits of space that supports 
  $\rank$, $\select$ and $\rangecount$ in $O(1 + \frac{\lg |U|}{\lg\lg |w|})$ time.
\end{theorem}

We also use a data structure for the \emph{Range Maximum Queries (RMQs)} over an integer array $V$.
Given an interval $[i..j]$ over $V$, a query $\rMq_V(i, j)$ returns a position in $[i..j]$ that has the maximum value in $V[i..j]$,
that is, $\rMq_V(i, j) = \arg\max_{k \in [i..j]}V[k]$.
We use the following result.
\begin{theorem}[\cite{2011FischerH_SpaceEfficPreprSchemFor}]\label{theo:rmq}
  For an integer array $V$ of length $n$, there is a data structure with $2n + o(n)$ bits of space that supports the RMQs in $O(1)$ time.
\end{theorem}

\subsection{FM-index}
The suffix array $\SA$ of $T$ is the integer array of length $n+1$ such that 
$\SA[i]$ is the starting position of the lexicographically $i$-th suffix of $T$.\footnote{Against convention, we include the empty string that starts with the position $n+1$ to $\SA$. In particular, $\SA[1] = n+1$ holds as the empty string is always the smallest suffix.}
We define the string $\Lstr$ (a.k.a.\ the \emph{Burrows-Wheeler Transform (BWT)}~\cite{Burrows1994BWT} of $T$) of length $n+1$ as follows:
\begin{equation*}
  \Lstr[i] =
  \begin{cases}
    \$           & (\SA[i] = 1),\\
    T[\SA[i]-1]  & (\SA[i] > 1).
  \end{cases}
\end{equation*}
We define the string $\Fstr$ of length $n+1$ as $\Fstr = T[\SA[1]] T[\SA[2]] \cdots T[\SA[n+1]]$.
The so-called \emph{LF-mapping} $\LF$ is the function defined to map a position $i$ to $j$ 
such that $\SA[j] = \SA[i] - 1$ (with the corner case $\LF(i) = 1$ for $\SA[i] = 1$).
A crucial point is that LF-mapping can be efficiently implemented by rank queries on $\Lstr$ and select queries on $\Fstr$
with $\LF(i) = \select_{\Fstr}(\rank_{\Lstr}(i, \Lstr[i]), \Lstr[i])$.
\footnote{In the plain LF-mapping, select queries on $\Fstr$ can be implemented by a simple table that counts, for each character $c$, the number of occurrences of characters smaller than $c$ in $T$, but it is not the case in our generalized LF-mapping for pal-matching.}
The occurrences of pattern $P$ in $T$ can be answered by finding the maximal interval $[P_b..P_e]$ in the $\SA$ array such that $T[\SA[i]..]$ 
is prefixed by $P$ iff $i \in [P_b..P_e]$, and computing the $\SA$-values in the interval.
For a string $w$ and character $c$, the so-called \emph{backward search} computes 
the maximal interval in the $\SA$ prefixed by $cw$ from that of $w$
using a similar mechanism of the LF-mapping (see~\cite{Ferragina2000ODS} for more details).

\section{Encodings for pal-matching}
The pal-matching algorithms in~\cite{I2013Ppm} are based on the $\lpal$-encoding of a string $w$, denoted as $\lpal_{w}$.
$\lpal_{w}$ is the integer array of length $|w|$ 
such that, for any position $1 \le i \le |w|$, $\lpal_{w}[i]$ is the length of the longest suffix-palindrome of $w[1..i]$.
See Table~\ref{tbl:lpal_ssp} for example.

\begin{lemma}[Lemma 2 in~\cite{I2013Ppm}]\label{lem:lpal_match}
  For any strings $x$ and $y$, $x$ and $y$ pal-match iff $\lpal_{x} = \lpal_{y}$.
\end{lemma}

Although Lemma~\ref{lem:lpal_match} is sufficient to design suffix-tree type indexes,
it seems that the $\lpal$-encoding is not suitable to design FM-index type indexes.
For example, more than one position could change when a character is prepended (see Table~\ref{tbl:lpal_ssp})
and this unstable property make messes up lexicographic order of $\lpal$-encoded suffixes,
which prevents us to implement LF-mapping space efficiently.

In this paper, we introduce a new encoding suitable to design FM-index type indexes for pal-matching.
Our new encoding is based on the shortest suffix-palindrome for each prefix,
where the shortest suffix is chosen excluding the trivial palindromes of length $\le 1$.
We call the encoding the \emph{shortest suffix-palindrome encoding} (the $\ssp$-encoding in short).
For any string $w$, the $\ssp$-encoding $\ssp_{w}$ of $w$ is the integer array of length $|w|$
such that, for any position $1 \le i \le |w|$, 
$\ssp_{w}[i]$ is the length of the non-trivial shortest suffix-palindrome of $w[..i]$
if such exists, and otherwise $\infty$. See Table~\ref{tbl:lpal_ssp} for example.

\begin{table*}[t]
  \caption{A comparison between $\lpal$ and $\ssp$ for $w = \idtt{abbbabb}$ and $w' = \idtt{b}w = \idtt{babbbabb}$.
    The values that change when prepending \idtt{b} to $w$ are underlined.}
  \label{tbl:lpal_ssp}
  \centering{
    \begin{tabular}{rcccccccc}
      $w = $          &          & \idtt{a} & \idtt{b} & \idtt{b} & \idtt{b} & \idtt{a}  & \idtt{b}  & \idtt{b} \\
      $\lpal_{w} = $   &          & 1        & 1        & 2        & 3        & 5         & 3         & 5 \\
      $\ssp_{w} = $    &          & $\infty$ & $\infty$ & 2        & 2        & 5         & 3         & 2 \\
      $w' = $         & \idtt{b} & \idtt{a} & \idtt{b}  & \idtt{b} & \idtt{b} & \idtt{a}  & \idtt{b}  & \idtt{b} \\
      $\lpal_{w'} = $  & 1        & 1        & \underline{3} & 2        & 3        & 5         & \underline{7} & 5 \\
      $\ssp_{w'} = $   & $\infty$ & $\infty$ & \underline{3} & 2        & 2        & 5         & 3         & 2 \\
    \end{tabular}
  }
\end{table*}

\begin{lemma}\label{lem:ssp_match}
  Two strings $x$ and $y$ pal-match iff $\ssp_{x} = \ssp_{y}$.
\end{lemma}
\begin{proof}
  Since the $\ssp$-encoding relies only on palindromic structures, the direction from left to right is clear.

  In what follows, we focus on the opposite direction; $x$ and $y$ pal-match if $\ssp_{x} = \ssp_{y}$.
  Assume for contrary that $x$ and $y$ does not pal-match.
  Without loss of generality, we can assume that there are positions $i$ and $j$ such that $x[i..j]$ is a palindrome but $y[i..j]$ is not,
  with smallest $j$ if there are many.
  Note that the smallest assumption on $j$ implies that $y[i+1..j-1]$ is a palindrome:
  If $y[i+1..j-1]$ is not a palindrome (clearly $|y[i+1..j-1]| > 1$ in such a case), 
  $j-1$ must be a smaller position that satisfies the above condition
  because $x[i+1..j-1]$ is a palindrome.
  Let $k = \ssp_{x}[j] = \ssp_{y}[j]$.
  Since $x[i..j]$ is a palindrome, it holds that $1 < k \le |x[i..j]|$.
  Moreover, $k \neq |y[i..j]|$ as $y[i..j]$ is not a palindrome.
  Since the palindrome $x[i..j]$ has a suffix-palindrome of length $k$, the prefix $x[i..i+k-1]$ of length $k$ is a palindrome, too.
  On the other hand, since $y[i..j]$ is not a palindrome that has a suffix-palindrome of length $k$, the prefix $y[i..i+k-1]$ of length $k$ cannot be a palindrome.
  This contradicts the smallest assumption on $j$ because $i+k-1$ is a smaller position such that $x[i..i+k-1]$ and $y[i..i+k-1]$ disagree on their palindromic structures.
\end{proof}

In contrast to the $\lpal$-encoding, the $\ssp$-encoding has a stable property when prepending a character.
\begin{lemma}\label{lem:ssp_change}
  For any string $w$ and character $c$, there is at most one position $i~(1 \le i \le |w|)$ such that $\ssp_{w}[i] \neq \ssp_{cw}[i+1]$.
  Moreover, if such a position $i$ exists, $\ssp_{w}[i] = \infty$ and $\ssp_{cw}[i+1] = i+1$.
\end{lemma}
\begin{proof}
  By definition it is obvious that $\ssp_{w}[i] = \ssp_{cw}[i+1]$ if $\ssp_{w}[i] \neq \infty$.
  In what follows, we assume for contrary that there exist two positions $i$ and $i'$ with $1 \le i < i' \le |w|$ such that 
  $\ssp_{w}[i] = \infty > \ssp_{cw}[i+1]$ and $\ssp_{w}[i'] = \infty > \ssp_{cw}[i'+1]$.
  Note that $\ssp_{cw}[i+1] = i+1$ and $\ssp_{cw}[i'+1] = i'+1$ by definition,
  and $(cw)[..i+1]$ and $(cw)[..i'+1]$ are palindromes.
  Since $(cw)[..i+1]$ is a prefix-palindrome of $(cw)[..i'+1]$, it is also a suffix-palindrome of $(cw)[..i'+1]$.
  It contradicts that $(cw)[..i'+1]$ is the non-trivial shortest suffix-palindrome of $(cw)[..i'+1]$.
\end{proof}

We consider yet another encoding based on the shortest suffix of $w[..i-1]$
that is extended outwards when appending a character $w[i]$.
The concept is closely related to the $\ssp$-encoding because the extended palindrome is 
the non-trivial shortest suffix-palindrome of $w[..i]$.
An advantage of this new encoding is that we can reduce the number of distinct integers to be used to $O(\min(\sigma, \lg |w|))$,
which will be used (in a symmetric way) to define $\palLstr$ and obtain a space-efficient FM-index specialized for pal-matching.

For any string $w$ we partition the suffix-palindromes (including the empty suffix) by the characters they have immediately to their left
and call each group a \emph{suffix-pal-group} for $w$.
We utilize the following lemma.
\begin{lemma}\label{lem:num_groups}
  For any string $w$, the number of suffix-pal-groups for $w$ is $O(\min(\sigma, \lg |w|))$.
\end{lemma}
\begin{proof}
  It is obvious that the number of suffix-pal-groups is at most $\sigma$ because each character is associated to at most one suffix-pal-group.
  Also it is known that the lengths of the suffix-palindromes can be represented by $O(\lg |w|)$ arithmetic progressions
  and each arithmetic progression induces a period in the involved suffix (e.g., see~\cite{2014ISIBT_ComputPalinFactorAndPalin}).
  Then we can see that every suffix-palindrome represented by an arithmetic progression is in the same group.
  Hence there are $O(\lg |w|)$ groups.
\end{proof}

The next lemma shows that pal-matching strings share the same structure of suffix-pal-groups.
\begin{lemma}\label{lem:same_groups}
  Let $x$ and $y$ be strings that pal-match and let $i$ and $j$ be integers with $1 \le i < j \le |x| = |y|$.
  If $x[i+1..]$ and $x[j+1..]$ are palindromes with $x[i] = x[j]$,
  then $y[i+1..]$ and $y[j+1..]$ are palindromes with $y[i] = y[j]$.
\end{lemma}
\begin{proof}
  Since the palindrome $x[i+1..]$ has a suffix-palindrome of length $k = |x[j+1..]|$,
  it also has a prefix-palindrome of length $k$, that is,
  $x[i+1..i+k]$ is a palindrome.
  Also, $x[i+k+1] = x[j]$ holds.
  Since $x[i] = x[j] = x[i+k+1]$, $x[i..i+k+1]$ is a palindrome.

  Since $x$ and $y$ pal-match, $y[i+1..]$, $y[j+1..]$ and $y[i..i+k+1]$ are palindromes.
  By transition of equivalence induced by the palindromes $y[i..i+k+1]$ and $y[i+1..]$, we can see that $y[i] = y[i+k+1] = y[j]$.
  Thus the claim holds.
\end{proof}

Let the shortest palindrome in a suffix-pal-group be the representative of the group.
We assign consecutive integer identifiers starting from $1$ to the suffix-pal-groups in increasing order of their representative's lengths.
See Figure~\ref{fig:palGroup1} for example.
\begin{figure}[t]
 \centering{
   \includegraphics[scale=0.5]{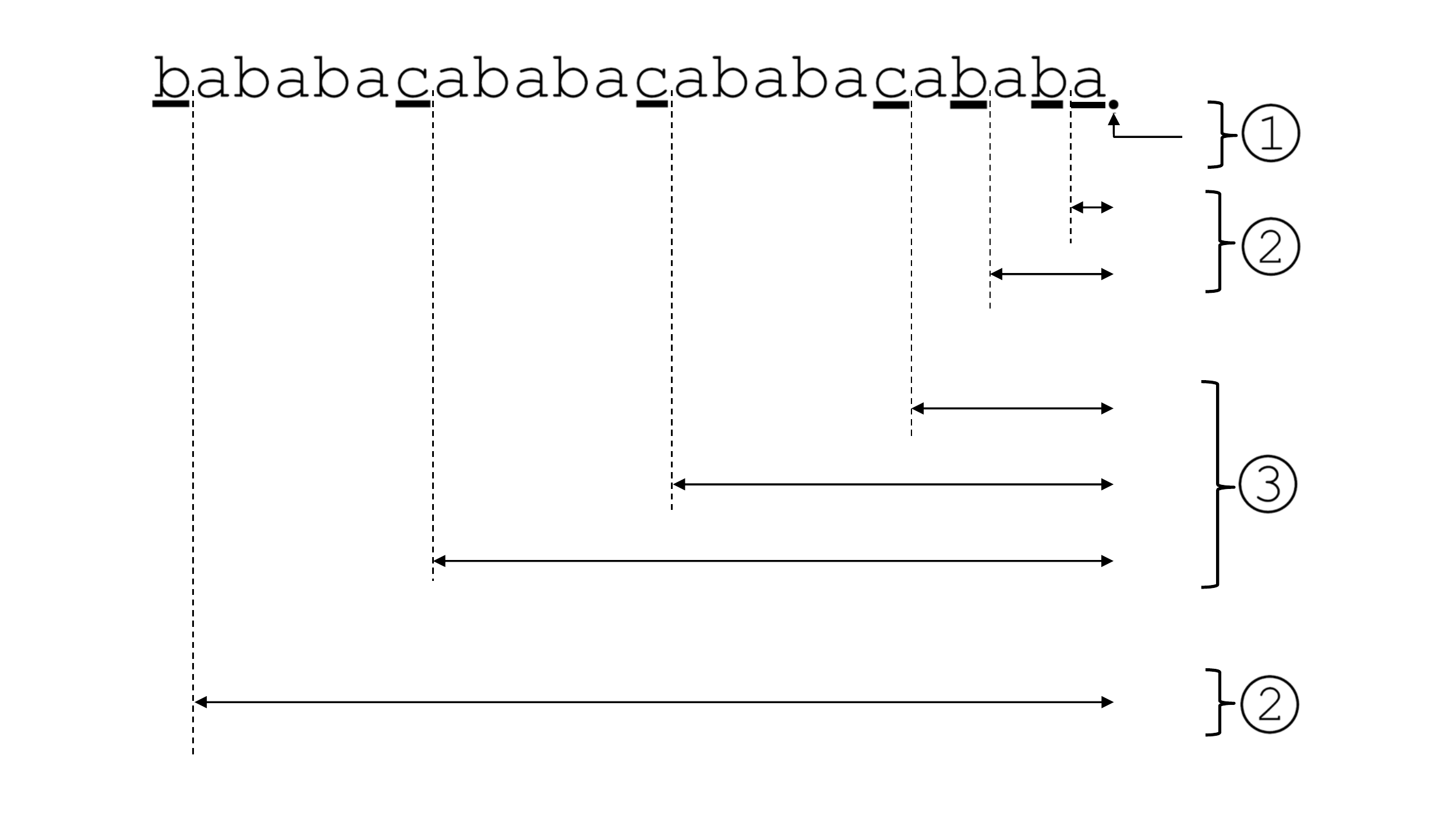}
 }
 \caption{An example of suffix-pal-groups for \idtt{bababababacababacababacababa}. The number enclosed in a circle denotes the pal-group-id.
 The suffix-palindromes in the suffix-pal-group with identifier $1$ (resp. $2$ and $3$) have \idtt{a} (resp. \idtt{b} and \idtt{c}) immediately to their left.
 The identifiers are given in increasing order of their representative's lengths, that is, $|\emptystr| = 0, |\idtt{a}| = 1$ and $|\idtt{ababa}| = 5$.
 }
 \label{fig:palGroup1}
\end{figure}

For any string $w$, we define the \emph{shortest suffix-pal-group encoding} $\sspg_{w}$ of $w$
as the integer array of length $|w|$ such that, for any position $1 \le i \le |w|$, 
$\sspg_{w}[i]$ is the identifier assigned to the suffix-pal-group of the suffix-palindrome in $w[..i-1]$ that is extended outwards by appending $w[i]$,
if such exists, and otherwise $\infty$. See Table~\ref{tbl:sspg} and Figure~\ref{fig:sspg_enc} for example.
Since the non-trivial shortest suffix of $w[..i]$ is extended outwards from the representative 
of the suffix-pal-group for $w[1..i-1]$ that has $w[i]$ immediately to the left,
$\sspg_{w}[i]$ has essentially equivalent information to $\ssp_{w}[i]$. Formally the next lemma holds.
\begin{lemma}\label{lem:ssp_sspg}
  For any string $x$ of length $k$, suppose we have the set of lengths of the representatives of suffix-pal-gropus of $x[..k-1]$.
  Given $\sspg_{x}[k]$ we can identify $\ssp_{x}[k]$, and vice versa.
\end{lemma}
\begin{proof}
  It is clear that $\ssp_{x}[k] = \infty$ iff $\sspg_{x}[k] = \infty$.
  Given $\sspg_{x}[k] \neq \infty$ we can identify $\ssp_{x}[k]$ from the representative of the suffix-pal-group with identifier $\sspg_{x}[k]$.
  Given $\ssp_{x}[k] \neq \infty$ we can identify $\sspg_{x}[k]$ from the representative that has length $\ssp_{x}[k] - 2$.
\end{proof}

The next lemma shows that the $\sspg$-encoding is another encoding for pal-matching,
and induces the same lexicographic order with the $\ssp$-encoding.
\begin{lemma}\label{lem:sspg_match_order}
  Let $x$ and $y$ be strings of length $k$ such that $\ssp_{x}[..k-1] = \ssp_{y}[..k-1]$.
  Then, $\ssp_{x}[k] = \ssp_{y}[k]$ iff $\sspg_{x}[k] = \sspg_{y}[k]$.
  Also, $\ssp_{x}[k] < \ssp_{y}[k]$ iff $\sspg_{x}[k] < \sspg_{y}[k]$.
\end{lemma}
\begin{proof}
  It follows from Lemma~\ref{lem:same_groups} that $x[..k-1]$ and $y[..k-1]$ have the same structure of suffix-pal-groups.
  By Lemma~\ref{lem:ssp_sspg}, $\ssp_x[k] = \ssp_y[k]$ if $\sspg_{x}[k] = \sspg_{y}[k]$, and vice versa.
  Since the identifiers of suffix-pal-groups are given in increasing order of their representative's lengths,
  it holds that $\ssp_{x}[k] < \ssp_{y}[k]$ if and only if $\sspg_{x}[k] < \sspg_{y}[k]$.
\end{proof}

\begin{table*}[t]
  \caption{A comparison between $\ssp_{w}$ and $\sspg_{w}$ for $w = \idtt{babbbabb}$.
  $\ssp_{w}[6] = 5$ because the non-trivial shortest suffix-palindrome of $w[1..6] = \idtt{babbba}$ is $\idtt{abbba}$, which is of length $5$.
  On the other hand, $\sspg_{w}[6] = 2$ because the shortest suffix-palindrome $\idtt{abbba}$ ending at $6$ is extended from $\idtt{bbb}$ 
  and the suffix-pal-group to which $\idtt{bbb}$ belongs for $w[1..5] = \idtt{babbb}$ has the identifier $2$.}
  \label{tbl:sspg}
  \centering{
    \begin{tabular}{rcccccccc}
      $w = $         & \idtt{b} & \idtt{a} & \idtt{b}  & \idtt{b} & \idtt{b} & \idtt{a}  & \idtt{b}  & \idtt{b} \\
      $\ssp_{w} = $   & $\infty$ & $\infty$ & 3 & 2        & 2        & 5         & 3         & 2 \\
      $\sspg_{w} = $  & $\infty$ & $\infty$ & 2 & 1        & 1        & 2         & 2         & 2 \\
    \end{tabular}
  }
\end{table*}

\begin{figure}[t]
 \centering{
   \includegraphics[scale=0.45]{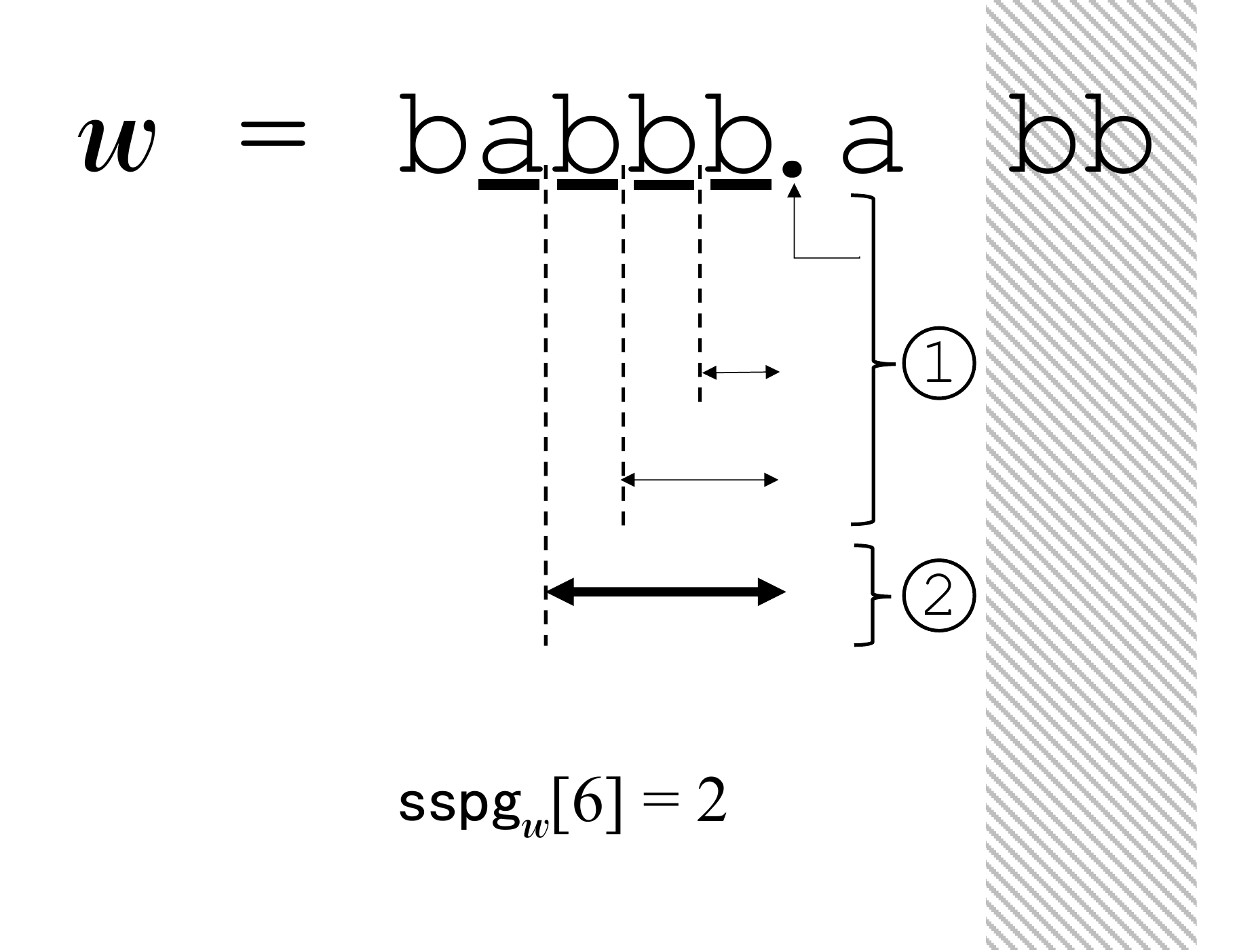}
 }
 \caption{Illustration to show $\sspg_{w}[6] = 2$ for $w = \idtt{babbbabb}$.
 }
 \label{fig:sspg_enc}
\end{figure}

For any string $w$, let $\sppg(w) = \sspg_{\rev{w}}[|w|]$.
Intuitively, $\sppg(w)$ holds the information from which prefix-palindrome of $w[2..]$
the non-trivial shortest prefix-palindrome of $w$ is extended, and the information is 
encoded with the identifier defined in the completely symmetric way as the case of the suffix-pal-groups.
The function $\sppg(\cdot)$ will be applied to the suffixes of $T$ to define $\palFstr$ and $\palLstr$,
and the next lemma is a key to implement LF-mapping for our PalFM-index.
\begin{lemma}\label{lem:sspg_stable}
  Let $x$ and $y$ be strings of length $\ge 1$ such that $\sppg(x) = \sppg(y)$.
  Then, $\ssp_x \prec \ssp_y$ iff $\ssp_{x[2..]} \prec \ssp_{y[2..]}$.
\end{lemma}
\begin{proof}
  Let $i$ be the largest integer such that $x[2..i]$ and $y[2..i]$ pal-match.
  Since $\sppg(x) = \sppg(y)$, using Lemma~\ref{lem:sspg_match_order} in a symmetric way,
  it holds that $x[..i]$ and $y[..i]$ pal-match.
  Recall Lemma~\ref{lem:ssp_change} that at most one $\infty$ in $\ssp_{x[2..]}$ (resp. $\ssp_{y[2..]}$) 
  turns into the largest possible integer at the changed position when prepending $x[1]$ (resp. $y[1])$.
  We analyze the cases focusing on the changed positions:
  \begin{enumerate}
    \item The claim clearly holds if neither $\ssp_x$ nor $\ssp_y$ has the changed position less than or equal to $i+1$.
    \item If both of $\ssp_x$ and $\ssp_y$ have the changed position at $j \le i+1$,
      it holds that $\ssp_x[j] = \ssp_y[j] = j$ and $\ssp_{x[2..]}[j-1] = \ssp_{y[2..]}[j-1] = \infty$,
      which also indicates that $j < i+1$.
      Since this change does not affect the lexicographic order, the claim holds.
      See the left part of Figure~\ref{fig:sspg_stable} for an illustration of this case.
    \item Assume $\ssp_y$ has the changed position at $j \le i+1$, but $\ssp_x$ does not.
      Since $x[..i]$ and $y[..i]$ pal-match, $j$ cannot be less than $i+1$, and hence, $j = i+1$ and
      $\ssp_x[i+1] = \ssp_{x[2..]}[i] \prec i+1 = \ssp_y[i+1] \prec \infty = \ssp_{y[2..]}[i]$.
      Note that the lexicographic order between $\ssp_x$ and $\ssp_y$ (resp. $\ssp_{x[2..]}$ and $\ssp_{y[2..]}$) is determined by
      that between $\ssp_x[i+1]$ and $\ssp_y[i+1]$ (resp. $\ssp_{x[2..]}[i]$ and $\ssp_{y[2..]}[i]$).
      Since the lexicographic order between $\ssp_x[i+1]$ and $\ssp_y[i+1]$ is the same as that between $\ssp_{x[2..]}[i]$ and $\ssp_{y[2..]}[i]$,
      the claim holds.
      See the right part of Figure~\ref{fig:sspg_stable} for an illustration of this case.
  \end{enumerate}
  Thus, we conclude that the lemma holds.
\end{proof}

\begin{figure}[t]
 \centering{
   \includegraphics[scale=0.45]{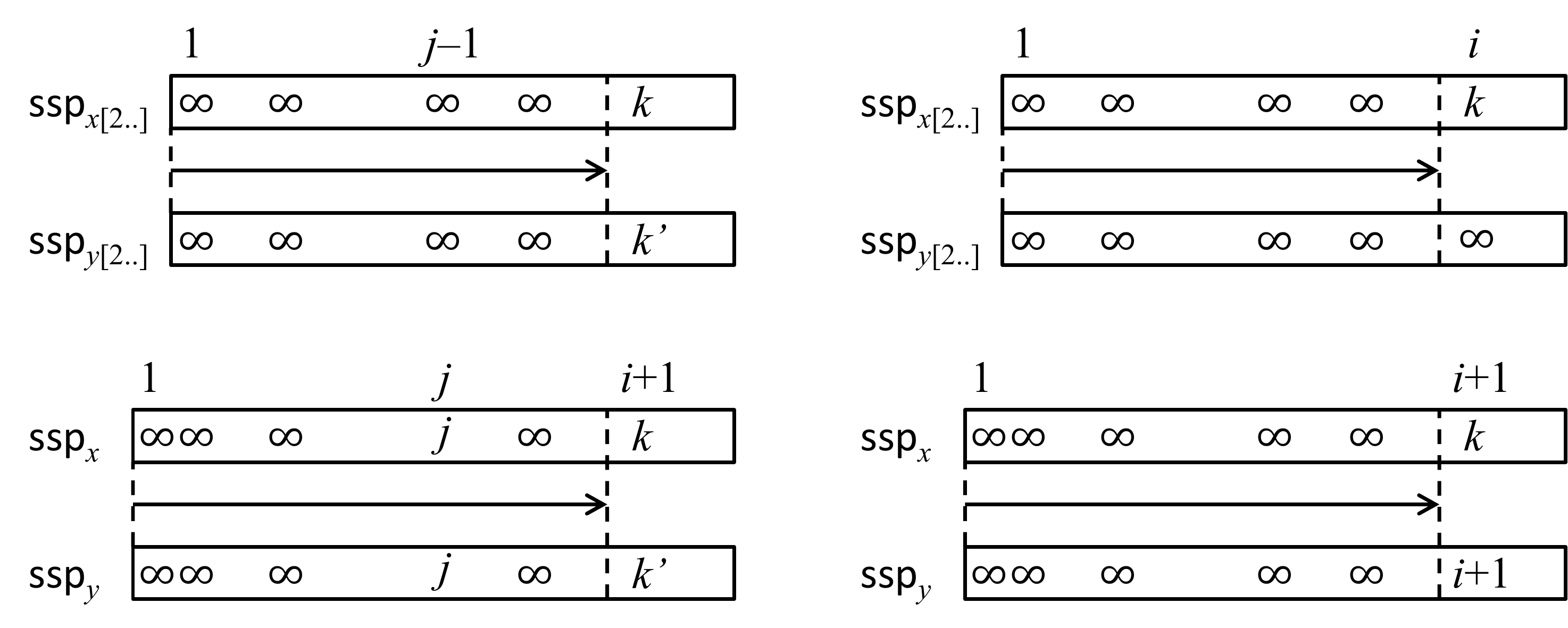}
 }
 \caption{The left (resp. right) figure illustrates the second (resp. third) case in the proof of Lemma~\ref{lem:sspg_stable}.
 }
 \label{fig:sspg_stable}
\end{figure}

\section{Computational results for new encodings}
In this section, we show that the $\ssp$- and $\sspg$-encodings can be computed in linear time for a given string.

We use the following known results.
\begin{lemma}[\cite{1975Manacher_NewLinearTimeOnLine}]\label{lem:Manacher}
  For any string $w$, we can compute all the maximal palindromes in $O(|w|)$ time.
\end{lemma}
\begin{lemma}[Lemma 3 in~\cite{I2013Ppm}]\label{lem:compute_lpal}
  For any string $w$, we can compute $\lpal_{w}$ in $O(|w|)$ time.
\end{lemma}

Using Lemmas~\ref{lem:Manacher} and~\ref{lem:compute_lpal}, we obtain:
\begin{lemma}\label{lem:compute_ssp}
  For any string $w$, we can compute $\ssp_{w}$ in $O(|w|)$ time.
\end{lemma}
\begin{proof}
  Manacher's algorithm~\cite{1975Manacher_NewLinearTimeOnLine} can compute the radius of the maximal palindrome in increasing order of centers in linear time.
  It can be extended to compute the length $\lpal_{w}[i]$ of the longest palindrome ending at each position $i$ 
  because the maximal palindrome with the smallest center that ends at position $\ge i$ 
  gives us the longest suffix-palindrome ending at $i$ by truncating the palindrome at $i$ (e.g., see Lemma 3 of~\cite{I2013Ppm}).
  In a similar way, we can compute the length $\lpal'_{w}[i]$ of the second longest palindrome ending at $i$.

  Using $\lpal_{w}$ and $\lpal'_{w}$, we can compute $\ssp_{w}[i]$ in increasing order as follows:
  \begin{enumerate}
    \item If $\lpal_{w}[i] = 1$, then $\ssp_{w}[i] = \infty$.
    \item If $\lpal_{w}[i] > 1$ and $\lpal'_{w}[i] = 1$, then $\ssp_{w}[i] = \lpal_{w}[i]$.
    \item If $\lpal_{w}[i] > 1$ and $\lpal'_{w}[i] > 1$, then $\ssp_{w}[i] = \ssp_{w}[i - \lpal_{w}[i] + \lpal'_{w}[i]]$.
  \end{enumerate}
  In the third case, we use the fact that the non-trivial shortest suffix-palindrome ending at $i$ has length $\le \lpal'_{w}[i]$ and 
  it ends at $i - \lpal_{w}[i] + \lpal'_{w}[i]$, too.

  Clearly all can be done in $O(|w|)$ time.
\end{proof}

For any string $w$, let $\nsspg_w$ denote the array of length $|w|$ such that
$\nsspg_w[i]$ stores the number of suffix-pal-groups for $w[..i]$.
\begin{lemma}\label{lem:compute_nsspg}
  For any string $w$, we can compute $\nsspg_w$ in $O(|w|)$ time.
\end{lemma}
\begin{proof}
  Let $\spp_w$ be the array defined in a symmetric way of $\ssp_w$
  such that $\spp_w[i]$ stores the length of the non-trivial shortest prefix-palindrome starting at $i$ (or $\infty$ if such a palindrome does not exist).
  Using Lemma~\ref{lem:compute_ssp} in a symmetric way, we can compute $\spp_w$ in $O(|w|)$ time.

  Let us focus on the palindromes involved in $\nsspg_w[j]$.
  First, there is a suffix-pal-group for $w[..j]$ that has $w[j+1]$ immediately to their left iff $\lpal_w[j+1] > 1$.
  Next observe that the palindromes in other suffix-pal-groups for $w[..j]$, which do not have $w[j+1]$ immediately to their left, are the maximal palindromes ending at $j$.
  Also, a maximal palindrome $w[i..j]$ is the representative (i.e., the shortest palindrome) in a suffix-pal-group to which it belongs.
  if and only if $\spp_w[i-1] > |w[i-1..j]|$ or $i = 1$. See Figure~\ref{fig:nsspg_obs} for illustrations of these observations.

  Based on the above observations, we compute $\nsspg_w$ as follows:
  First, we compute the maximal palindromes and $\lpal_w$ in $O(|w|)$ time by Lemmas~\ref{lem:Manacher} and~\ref{lem:compute_lpal}.
  Next we check every maximal palindrome and assign it to its ending position if it is a representative, which can be done in $O(|w|)$ time in total.
  We also check if $\lpal_w[j+1] > 1$ for all positions $j$ in $O(|w|)$ time to count a suffix-pal-group that has $w[j+1]$ immediately to their left.
  To sum up, $\nsspg_w$ can be computed in $O(|w|)$ time.
\end{proof}

\begin{figure}[t]
 \centering{
   \includegraphics[scale=0.45]{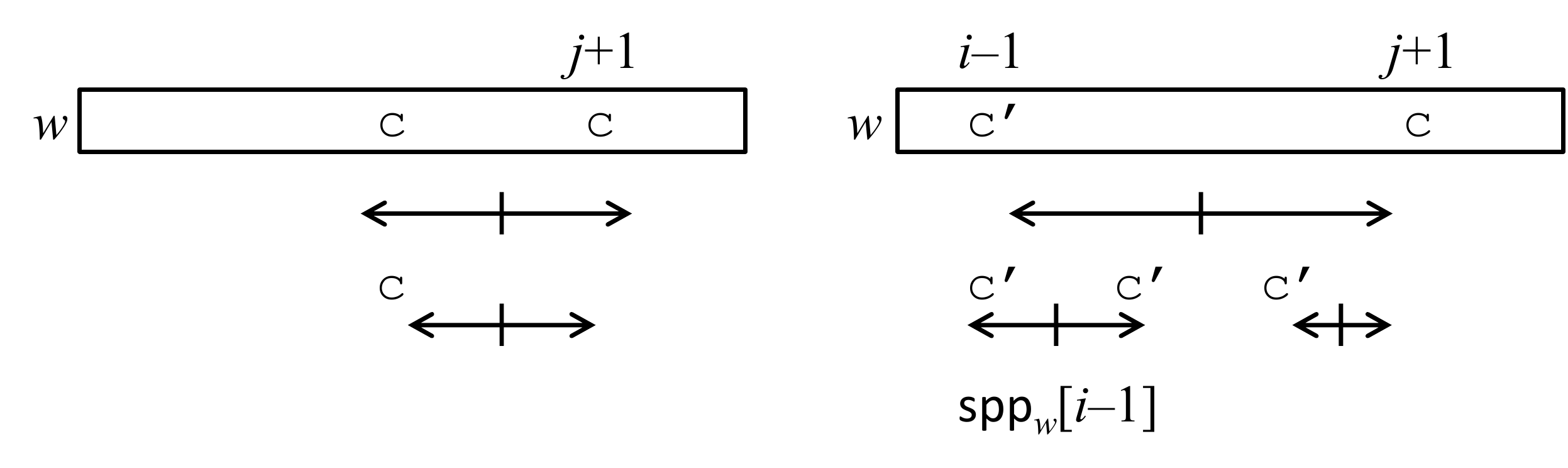}
 }
 \caption{The left figure illustrates the case with $\lpal_w[j+1] > 1$, in which we see that there is a suffix-pal-group for $w[..j]$ that has $w[j+1] = \idtt{c}$ immediately to their left. The right figure illustrates the case with $\spp_w[i-1] \le |w[i-1..j]|$, in which we see that the maximal palindrome $w[i..j]$ is not the representative because there is a shorter palindrome that ends at $j$ and has the same character $\idtt{c}'$ immediately to the left.
 }
 \label{fig:nsspg_obs}
\end{figure}

Generalizing the algorithm presented in the proof of Lemma~\ref{lem:compute_nsspg}, we obtain:
\begin{lemma}\label{lem:compute_sspg}
  For any string $w$, we can compute $\sspg_{w}$ in $O(|w|)$ time.
\end{lemma}
\begin{proof}
  We modify the algorithm presented in the proof of Lemma~\ref{lem:compute_nsspg} slightly.
  Now the task is to count, for every position $j+1$, the number of suffix-pal-groups for $w[..j]$ 
  whose representative is shorter than $\ssp[j+1] - 1$ because the number is exactly $\sspg_{w}[j+1]$ by definition.
  We check every maximal palindrome $w[i..j]$
  and assign it to its ending position $j$ if $\spp_w[i-1] > |w[i-1..j]|$ and $\ssp[j+1] - 1 > j - i + 1$.
  Finally the number of representatives assigned to $j$ plus one is $\sspg_{w}[j+1]$.
  Similarly to the proof of Lemma~\ref{lem:compute_nsspg}, all can be done in $O(|w|)$ time.
\end{proof}

\section{PalFM-index}
The PalFM-index of $T$ conceptually sort the suffixes of $T$ in lexicographic order of their $\ssp$-encodings (or equivalently $\sspg$-encodings).
Let $\palSA$ be the integer array of length $n+1$ such that
$\palSA[i]$ is the starting position of the $i$-th suffix of $T$ in $\ssp$-encoded order.
We define the strings $\palFstr$ and $\palLstr$ of length $n+1$ based on $\sppg$ function applied to the sorted suffixes.
Formally, for any position $i~(1 \le i \le n+1)$ we define:
\begin{equation*}
  \palFstr[i] =
  \begin{cases}
    \$                      & \text{if $i = 1$,}\\
    \sppg(T[\palSA[i]..])   & \text{otherwise.}
  \end{cases}
\end{equation*}
\begin{equation*}
  \palLstr[i] =
  \begin{cases}
    \$                      & \text{if $\palSA[i] = 1$,}\\
    \sppg(T[\palSA[i]-1..]) & \text{otherwise.}
  \end{cases}
\end{equation*}
See Figure~\ref{fig:palT} for example.

As in the case of $\LF$, we define a function $\palLF: i \mapsto j$
so that $\palSA[j] = \palSA[i] - 1$ (with the corner case $\palLF(i) = 1$ for $\palSA[i] = 1$).
Thanks to Lemma~\ref{lem:sspg_stable}, for any value $c$, 
the suffixes used to obtain $i$-th $k$ in $\palLstr$ and in $\palFstr$ are the same,
which enables us to implement the $\palLF$ function by
$\palLF(i) = \select_{\palFstr}(\rank_{\palLstr}(i, \palLstr[i]), \palLstr[i])$.
See Figure~\ref{fig:LFmap} for an illustration.

\begin{figure}[t]
  \centering{
    \begin{tabular}{|c|l|l|l|c|c|c|c|}
      \hline
      $i$ & $T[i..]$           & $\ssp_{T[i..]}$                    & $\ssp_{T[\palSA[i]..]}$            & $\palSA[i]$ & $\palFstr[i]$ & $\palLstr[i]$  & $\palLF(i)$\\ \hline
      1   & \idtt{abbabbcbc} & \idtt{\infty\infty2432\infty33} & $\emptystr$                     & 10          & \idtt{\$}     & \idtt{\infty} & 2  \\
      2   & \idtt{bbabbcbc}  & \idtt{\infty2\infty32\infty33}  & \idtt{\infty}                   & 9           & \idtt{\infty} & \idtt{\infty} & 5  \\
      3   & \idtt{babbcbc}   & \idtt{\infty\infty32\infty33}   & \idtt{\infty2\infty32\infty33}  & 2           & \idtt{1}      & \idtt{2}      & 6  \\
      4   & \idtt{abbcbc}    & \idtt{\infty\infty2\infty33}    & \idtt{\infty2\infty33}          & 5           & \idtt{1}      & \idtt{\infty} & 7  \\
      5   & \idtt{bbcbc}     & \idtt{\infty2\infty33}          & \idtt{\infty\infty}             & 8           & \idtt{\infty} & \idtt{2}      & 8  \\
      6   & \idtt{bcbc}      & \idtt{\infty\infty33}           & \idtt{\infty\infty2432\infty33} & 1           & \idtt{2}      & \idtt{\$}     & 1  \\
      7   & \idtt{cbc}       & \idtt{\infty\infty3}            & \idtt{\infty\infty2\infty33}    & 4           & \idtt{\infty} & \idtt{2}      & 9  \\
      8   & \idtt{bc}        & \idtt{\infty\infty}             & \idtt{\infty\infty3}            & 7           & \idtt{2}      & \idtt{2}      & 10 \\
      9   & \idtt{c}         & \idtt{\infty}                   & \idtt{\infty\infty32\infty33}   & 3           & \idtt{2}      & \idtt{1}      & 3  \\
      10  & $\emptystr$      & $\emptystr$                     & \idtt{\infty\infty33}           & 6           & \idtt{2}      & \idtt{1}      & 4  \\ \hline
    \end{tabular}
  }
  \caption{An example of $\palSA[i]$, $\palFstr[i]$ and $\palLstr[i]$ for $T = \idtt{abbabbcbc}$.}
  \label{fig:palT}
\end{figure}

\begin{figure}[t]
 \centering{
   \includegraphics[scale=0.45]{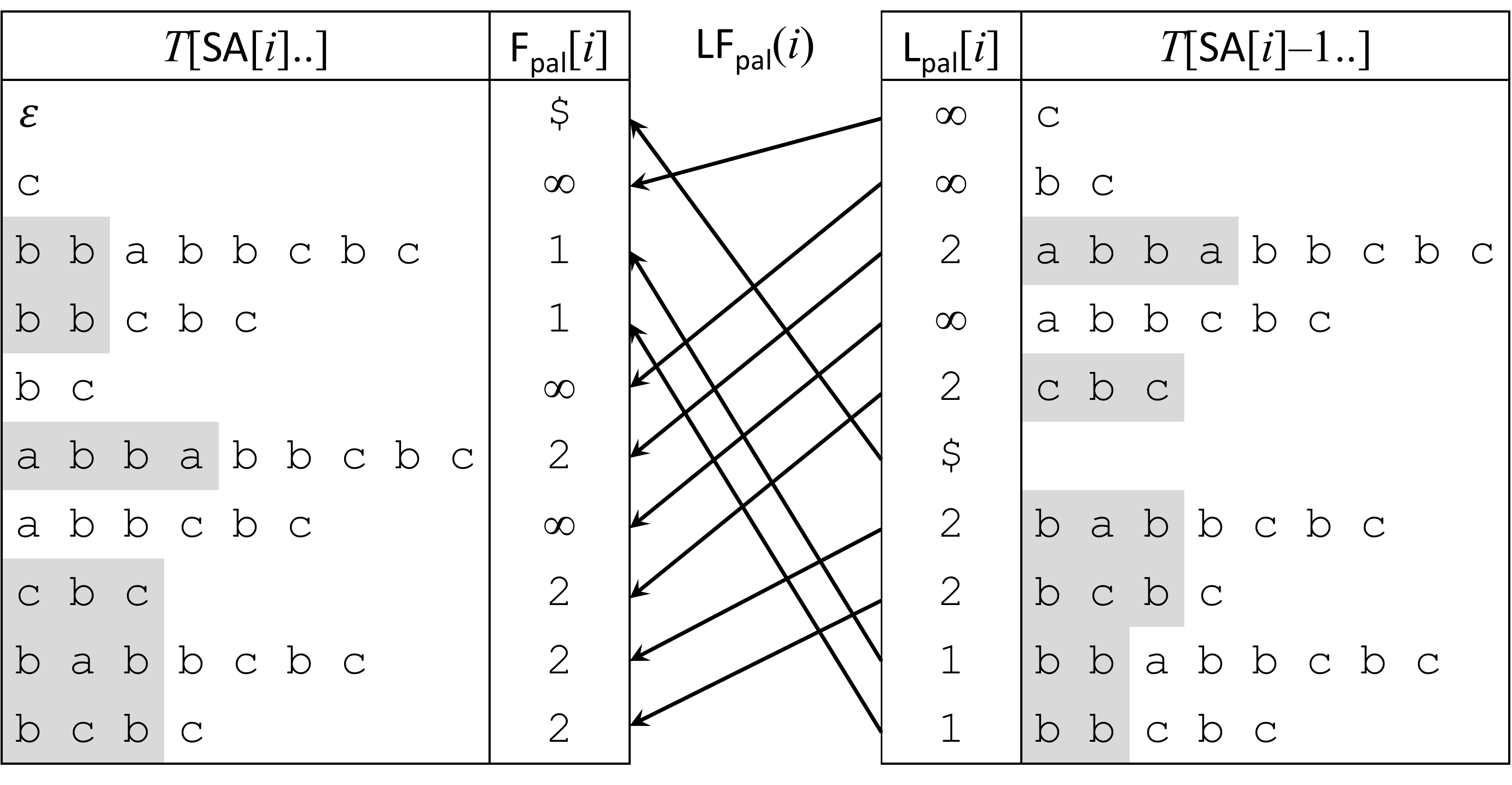}
 }
 \caption{An illustration for $\palFstr[i]$, $\palLstr[i]$ and $\palLF(i)$.
   Except the corner cases that have \idtt{\$}, $\palFstr[i]$ and $\palLstr[i]$ are defined by $\sppg(T[\palSA[i]..])$ and $\sppg(T[\palSA[i]-1..])$, respectively.
   Since $\sppg(w)$ encodes the information about the non-trivial shortest prefix of $w$,
   in each row the non-trivial shortest prefix is shown in grayed background.
   For example, $\sppg(\idtt{abbabbcbc}) = \idtt{2}$ because its non-trivial shortest prefix-palindrome \idtt{abba}
   is extended from the prefix-palindrome \idtt{bb} of \idtt{bbabbcbc} and \idtt{bb} belongs to 
   the prefix-pal-group with the identifier $\idtt{2}$.
   Observe that $\palFstr$ is a permutation of $\palLstr$ since
   both $\palFstr$ and $\palLstr$ use every suffix $w$ of $T$ exactly once to obtain $\sppg(w)$.
   Roughly speaking, $\palLF(\cdot)$ is meant to map a row having a suffix $w$ in the $T[\palSA[i]-1..])$ column
   to the row having the same suffix $w$ in the $T[\palSA[i]..]$ column.
   Thanks to Lemma~\ref{lem:sspg_stable}, for any value $k$, 
   the suffixes used to obtain $i$-th $k$ in $\palLstr$ and in $\palFstr$ are the same,
   and hence, one can observe visually that the arrows starting from the same $\palLstr$-value are not crossed.
 }
 \label{fig:LFmap}
\end{figure}

For any string $w$, let $w$-interval refer to the maximal interval $[b..e]$ such that $\ssp_{T[\palSA[i]..]}$ is prefixed by $\ssp_{w}$,
where $w$-interval is empty if there is no substring of $T$ that pal-matches with $w$.
Notice that the substring of $T$ of length $|w|$ starting at $\palSA[i]$ pal-matches with $w$ iff $i \in [b..e]$.
A single step of backward search computes $cw$-interval from $w$-interval for some character $c$.

The following theorems are the main contributions of this paper.
\begin{theorem}\label{theo:pal_counting}
  Let $T$ be a string of length $n$ over an alphabet of size $\sigma$.
  There is a data structure of $2n \lg \min(\sigma, \lg n) + 2n + o(n)$ bits of space 
  to support the counting queries for the pal-matching problem in $O(m)$ time,
  where $m$ is the length of a given pattern $P$.
\end{theorem}
\begin{proof}
  We use the data structures of Theorem~\ref{theo:wt} for $\palLstr$ and $\palFstr$,
  and the RMQ data structure of Theorem~\ref{theo:rmq} for the integer array $V$ with $V[i] = \palLF(i)$.
  Since the number of distinct symbols in $\palLstr$ and $\palFstr$ are $O(\min(\sigma, \lg n))$ by Lemma~\ref{lem:num_groups},
  the data structures occupy $2n \lg \min(\sigma, \lg n) + 2n + o(n)$ bits of space in total
  and all queries ($\rank$, $\select$, $\rangecount$ and $\rMq$) can be supported in $O(1)$ time.

  The number of occurrences of $P$ can be answered by computing the width of $P$-interval.
  Thus we focus on a single step of backward search.
  In a general setting, for any string $w$ and a character $c$,
  we show how to compute $cw$-interval $[b'..e']$ in $O(1)$ time from $w$-interval $[b..e]$, $\sppg(cw)$ and the number $g$ of prefix-pal-groups of $w$.
  The procedure differs depending on $\sppg(cw) = \infty$ or not.
  \begin{enumerate}
    \item When $\sppg(cw) = k \neq \infty$.
      Using Lemma~\ref{lem:sspg_match_order} in a symmetric way, $[b'..e']$ is obtained by mapping the positions of $\sppg(cw)$ in $\palLstr[b..e]$ by the $\palLF$ function.
      More specifically, $b' = \select_{\palFstr}(\rank_{\palLstr}(b - 1, k) + 1, k)$ and 
      $e' = \select_{\palFstr}(\rank_{\palLstr}(e, k), k)$, which can be computed in $O(1)$ time.
    \item When $\sppg(cw) = \infty$. We note that $[b'..e']$ is the maximal interval such that 
      $T[\palSA[i]..]$ does not have non-trivial prefix-palindrome (i.e. $\sppg(T[\palSA[i]..]) = \infty$) or 
      $T[\palSA[i]..]$ has the non-trivial shortest prefix-palindrome of length longer than $|cw|$ (i.e. $\sppg(T[\palSA[i]..]) > g$).
      Thus, $e' - b' + 1$ is equivalent to the number of occurrences of values larger than $g$ in $\palLstr[b..e]$,
      which can be computed in $\rangecount_{\palLstr}(b, e, g, \infty)$ in $O(1)$ time.
      Moreover, it holds that $e' = \palLF(\rMq_{V}(b, e))$ because 
      $\ssp(T[\palSA[i]-1..])$ with $\sppg(T[\palSA[i]-1..]) = \palLstr[i] > g$ is always lexicographically larger than 
      $\ssp(T[\palSA[j]-1..])$ with $\sppg(T[\palSA[j]-1..]) = \palLstr[j] \le g$.
      Thus, we can compute $[b'..e']$ in $O(1)$ time.
  \end{enumerate}

  Backward search for $P$ requires $\sppg(P[i..])$ and the number $g$ of prefix-pal-groups of $P[i..]$ for all $1 \le i \le m$,
  which can be computed by $\sspg_{\rev{P}}$ and $\nsspg_{\rev{P}}$ in $O(m)$ time using Lemmas~\ref{lem:compute_sspg} and~\ref{lem:compute_nsspg}.

  Putting all together, we get the theorem.
\end{proof}

\begin{theorem}\label{theo:pal_locating}
  Let $T$ be a string of length $n$ over an alphabet of size $\sigma$ and $\Delta$ be an integer in $[1..n]$.
  There is a data structure of $2n \lg \min(\sigma, \lg n) + \frac{n}{\Delta} \lg n + 3n + o(n)$ bits of space
  to support the locating queries for the pal-matching problem in $O(m + \Delta \occ)$ time,
  where $m$ is the length of a given pattern $P$ and $\occ$ is the number of occurrences to report.
\end{theorem}
\begin{proof}
  We use the data structure and the algorithm of Theorem~\ref{theo:pal_counting} to compute 
  $P$-interval in $2n (1 + \lg \min(\sigma, \lg n)) + o(n)$ bits of space and $O(m)$ time.
  The occurrences of $P$ (in the sense of pal-matching) can be answered by the $\palSA$-values in $P$-interval.
  We employ exactly the same sampling technique used in the FM-index to retrieve $\SA$-values (e.g., see~\cite{Ferragina2000ODS}):
  We make a bit vector $B$ of length $n+1$ marking the positions $i$ in $\palSA$ such that $\palSA[i] = \Delta k + 1$ for some integer $k$,
  and the sparse suffix array $S$ holding only the marked $\palSA$-values in the order.
  $B$ is equipped with a data structure to support the rank queries and the additional space to Theorem~\ref{theo:pal_counting}
  is $\frac{n}{\Delta} \lg n + n + o(n)$ bits in total.

  If position $i$ is marked, $\palSA[i]$ is retrieved by $S[\rank_B(i, 1)]$ in $O(1)$ time.
  If position $i$ is not marked, we apply LF-mapping $k$ times from $i$ until we reach a marked position $j$
  and retrieve $\palSA[i]$ by $S[\rank_B(j, 1)] + k$.
  Since text positions are marked every $\Delta$ positions, the number $k$ of LF-mapping steps is at most $\Delta$,
  and hence, $\palSA[i]$ can be retrieved in $O(\Delta)$ time.
  Therefore we can report each occurrence of $P$ in $O(\Delta)$ time, and the theorem follows.
\end{proof}

\section{Conclusions and future work}
In this paper, we developed new encoding schemes for pal-matching and proposed the PalFM-index, a space-efficient index for pal-matching based on the FM-index.
Future work includes to present an efficient construction algorithm of the PalFM-index, and
to reduce the space requirement (e.g.\ by incorporating with the idea of~\cite{2017GangulyST_PbwtAchievSuccinDataStruc_SODA}).
Another interesting research direction would be to develop a general framework to design FM-index type indexes in generalized pattern matching.
We believe that switching encoding from $\lpal$ to $\ssp$ to design the PalFM-indexes gives a good hint to pursue this direction,
and conjecture that any generalized pattern matching under a substring consistent equivalent relation~\cite{2016MatsuokaAIBT_GenerPatterMatchAndPeriod_TCS}
admits such shortest positional encodings to design FM-index type indexes.

\section*{Acknowledgements}
This work was supported by JSPS KAKENHI (Grant Number 19K20213).

\bibliography{refs}

\begin{thebibliography}{10}

\bibitem{2003AlloucheBCD_PalinCompl}
Jean-Paul Allouche, Michael Baake, Julien Cassaigne, and David Damanik.
\newblock Palindrome complexity.
\newblock {\em Theor. Comput. Sci.}, 292(1):9--31, 2003.

\bibitem{2010AnisiuAK_TotalPalinComplOfFinit}
Mira-Cristiana Anisiu, Valeriu Anisiu, and Zolt{\'a}n K{\'a}sa.
\newblock Total palindrome complexity of finite words.
\newblock {\em Discrete Mathematics}, 310(1):109--114, 2010.
\newblock \href {https://doi.org/http://dx.doi.org/10.1016/j.disc.2009.08.002}
  {\path{doi:http://dx.doi.org/10.1016/j.disc.2009.08.002}}.

\bibitem{2017BorozdinKRS_PalinLengtInLinearTime_CPM}
Kirill Borozdin, Dmitry Kosolobov, Mikhail Rubinchik, and Arseny~M. Shur.
\newblock Palindromic length in linear time.
\newblock In {\em Proc. 28th Annual Symposium on Combinatorial Pattern Matching
  ({CPM}) 2017}, pages 23:1--23:12, 2017.
\newblock \href {https://doi.org/10.4230/LIPIcs.CPM.2017.23}
  {\path{doi:10.4230/LIPIcs.CPM.2017.23}}.

\bibitem{2004BrlekHNR_PalinComplOfInfinWords}
Srecko Brlek, Sylvie Hamel, Maurice Nivat, and Christophe Reutenauer.
\newblock On the palindromic complexity of infinite words.
\newblock {\em Int. J. Found. Comput. Sci.}, 15(2):293--306, 2004.
\newblock \href {https://doi.org/10.1142/S012905410400242X}
  {\path{doi:10.1142/S012905410400242X}}.

\bibitem{Burrows1994BWT}
Michael Burrows and David~J Wheeler.
\newblock A block-sorting lossless data compression algorithm.
\newblock Technical report, HP Labs, 1994.

\bibitem{2001DroubayJP_EpistWordsAndSomeConst}
Xavier Droubay, Jacques Justin, and Giuseppe Pirillo.
\newblock Episturmian words and some constructions of de luca and rauzy.
\newblock {\em Theor. Comput. Sci.}, 255(1-2):539--553, 2001.
\newblock \href {https://doi.org/10.1016/S0304-3975(99)00320-5}
  {\path{doi:10.1016/S0304-3975(99)00320-5}}.

\bibitem{Ferragina2000ODS}
Paolo Ferragina and Giovanni Manzini.
\newblock Opportunistic data structures with applications.
\newblock In {\em FOCS}, pages 390--398, 2000.

\bibitem{2007FerraginaMMN_ComprRepresOfSequenAnd}
Paolo Ferragina, Giovanni Manzini, Veli M{\"{a}}kinen, and Gonzalo Navarro.
\newblock Compressed representations of sequences and full-text indexes.
\newblock {\em {ACM} Trans. Algorithms}, 3(2), 2007.

\bibitem{2014FiciGKK_SubquadAlgorForMinimPalin_JDA}
Gabriele Fici, Travis Gagie, Juha K{\"a}rkk{\"a}inen, and Dominik Kempa.
\newblock A subquadratic algorithm for minimum palindromic factorization.
\newblock {\em Journal of Discrete Algorithms}, 28:41--48, 2014.
\newblock StringMasters 2012 \& 2013 Special Issue (Volume 1).
\newblock URL:
  \url{http://www.sciencedirect.com/science/article/pii/S1570866714000525},
  \href {https://doi.org/https://doi.org/10.1016/j.jda.2014.08.001}
  {\path{doi:https://doi.org/10.1016/j.jda.2014.08.001}}.

\bibitem{2011FischerH_SpaceEfficPreprSchemFor}
Johannes Fischer and Volker Heun.
\newblock Space-efficient preprocessing schemes for range minimum queries on
  static arrays.
\newblock {\em SIAM J. Comput.}, 40(2):465--492, 2011.

\bibitem{2017GagieMV_EncodForOrderPreserMatch_ESA}
Travis Gagie, Giovanni Manzini, and Rossano Venturini.
\newblock An encoding for order-preserving matching.
\newblock In {\em Proc. 25th Annual European Symposium on Algorithms ({ESA})
  2017}, pages 38:1--38:15, 2017.
\newblock \href {https://doi.org/10.4230/LIPIcs.ESA.2017.38}
  {\path{doi:10.4230/LIPIcs.ESA.2017.38}}.

\bibitem{Galil1978LOR}
Zvi Galil and Joel~I. Seiferas.
\newblock A linear-time on-line recognition algorithm for ``palstar''.
\newblock {\em J. ACM}, 25(1):102--111, 1978.
\newblock \href {https://doi.org/http://doi.acm.org/10.1145/322047.322056}
  {\path{doi:http://doi.acm.org/10.1145/322047.322056}}.

\bibitem{2017GangulyST_PbwtAchievSuccinDataStruc_SODA}
Arnab Ganguly, Rahul Shah, and Sharma~V. Thankachan.
\newblock {pBWT}: Achieving succinct data structures for parameterized pattern
  matching and related problems.
\newblock In {\em Proc. 28th Annual {ACM-SIAM} Symposium on Discrete Algorithms
  ({SODA}) 2017}, pages 397--407, 2017.
\newblock \href {https://doi.org/10.1137/1.9781611974782.25}
  {\path{doi:10.1137/1.9781611974782.25}}.

\bibitem{Ganguly2017StructuralPatternMatching}
Arnab Ganguly, Rahul Shah, and Sharma~V. Thankachan.
\newblock Structural pattern matching - succinctly.
\newblock In {\em Proc. 28th International Symposium on Algorithms and
  Computation ({ISAAC}) 2017}, pages 35:1--35:13, 2017.
\newblock \href {https://doi.org/10.4230/LIPIcs.ISAAC.2017.35}
  {\path{doi:10.4230/LIPIcs.ISAAC.2017.35}}.

\bibitem{2009GlenJWZ_PalinRichn}
Amy Glen, Jacques Justin, Steve Widmer, and Luca~Q. Zamboni.
\newblock Palindromic richness.
\newblock {\em Eur. J. Comb.}, 30(2):510--531, 2009.
\newblock \href {https://doi.org/10.1016/j.ejc.2008.04.006}
  {\path{doi:10.1016/j.ejc.2008.04.006}}.

\bibitem{2008GolynskiRR_RedunOfSuccinDataStruc_SWAT}
Alexander Golynski, Rajeev Raman, and S.~Srinivasa Rao.
\newblock On the redundancy of succinct data structures.
\newblock In Joachim Gudmundsson, editor, {\em Proc. 11th Scandinavian Workshop
  on Algorithm Theory ({SWAT}) 2008}, volume 5124 of {\em Lecture Notes in
  Computer Science}, pages 148--159. Springer, 2008.

\bibitem{2003GrossiGV_HighOrderEntropComprText}
Roberto Grossi, Ankur Gupta, and Jeffrey~Scott Vitter.
\newblock High-order entropy-compressed text indexes.
\newblock In {\em Proc. 14th Annual {ACM-SIAM} Symposium on Discrete Algorithms
  ({SODA}) 2003}, pages 841--850. {ACM/SIAM}, 2003.

\bibitem{2010IIBT_CountAndVerifMaximPalin_SPIRE}
Tomohiro I, Shunsuke Inenaga, Hideo Bannai, and Masayuki Takeda.
\newblock Counting and verifying maximal palindromes.
\newblock In {\em Proc. 17th International Symposium on String Processing and
  Information Retrieval ({SPIRE}) 2010}, pages 135--146, 2010.

\bibitem{I2013Ppm}
Tomohiro I, Shunsuke Inenaga, and Masayuki Takeda.
\newblock Palindrome pattern matching.
\newblock {\em Theor. Comput. Sci.}, 483:162--170, 2013.
\newblock \href {https://doi.org/http://dx.doi.org/10.1016/j.tcs.2012.01.047}
  {\path{doi:http://dx.doi.org/10.1016/j.tcs.2012.01.047}}.

\bibitem{2014ISIBT_ComputPalinFactorAndPalin}
Tomohiro I, Shiho Sugimoto, Shunsuke Inenaga, Hideo Bannai, and Masayuki
  Takeda.
\newblock Computing palindromic factorizations and palindromic covers on-line.
\newblock In {\em Proc. 25th Annual Symposium on Combinatorial Pattern Matching
  ({CPM}) 2014}, volume 8486 of {\em Lecture Notes in Computer Science}, pages
  150--161. Springer, 2014.

\bibitem{1971TinocoUL_EstimOfSeconStrucIn}
Ignacio~Tinoco Jr., Olke~C. Uhlenbeck, and Mark~D. Levine.
\newblock Estimation of secondary structure in ribonucleic acids.
\newblock {\em Nature}, 230:362--367, 1971.

\bibitem{2021KimC_CompacIndexForCartesTree_CPM}
Sung{-}Hwan Kim and Hwan{-}Gue Cho.
\newblock A compact index for cartesian tree matching.
\newblock In Pawel Gawrychowski and Tatiana Starikovskaya, editors, {\em Proc.
  32nd Annual Symposium on Combinatorial Pattern Matching ({CPM}) 2021}, volume
  191 of {\em LIPIcs}, pages 18:1--18:19. Schloss Dagstuhl - Leibniz-Zentrum
  f{\"{u}}r Informatik, 2021.

\bibitem{2021KimC_SimplFmIndexForParam}
Sung{-}Hwan Kim and Hwan{-}Gue Cho.
\newblock Simpler {FM}-index for parameterized string matching.
\newblock {\em Inf. Process. Lett.}, 165:106026, 2021.
\newblock \href {https://doi.org/10.1016/j.ipl.2020.106026}
  {\path{doi:10.1016/j.ipl.2020.106026}}.

\bibitem{KMP77}
Donald~E. Knuth, James~H. Morris, and Vaughan~R. Pratt.
\newblock Fast pattern matching in strings.
\newblock {\em SIAM J. Comput.}, 6(2):323--350, 1977.

\bibitem{2015KosolobovRS_PalKIsLinearRecog_SOFSEM}
Dmitry Kosolobov, Mikhail Rubinchik, and Arseny~M. Shur.
\newblock Pal k is linear recognizable online.
\newblock In {\em Proc. 41st International Conference on Current Trends in
  Theory and Practice of Computer Science ({SOFSEM}) 2015}, pages 289--301,
  2015.

\bibitem{1975Manacher_NewLinearTimeOnLine}
Glenn~K. Manacher.
\newblock A new linear-time {``on-line''} algorithm for finding the smallest
  initial palindrome of a string.
\newblock {\em J. {ACM}}, 22(3):346--351, 1975.
\newblock URL: \url{http://doi.acm.org/10.1145/321892.321896}, \href
  {https://doi.org/10.1145/321892.321896} {\path{doi:10.1145/321892.321896}}.

\bibitem{2016MatsuokaAIBT_GenerPatterMatchAndPeriod_TCS}
Yoshiaki Matsuoka, Takahiro Aoki, Shunsuke Inenaga, Hideo Bannai, and Masayuki
  Takeda.
\newblock Generalized pattern matching and periodicity under substring
  consistent equivalence relations.
\newblock {\em Theor. Comput. Sci.}, 656:225--233, 2016.

\bibitem{2009RestivoR_BurrowWheelTransAndPalin}
Antonio Restivo and Giovanna Rosone.
\newblock Burrows-wheeler transform and palindromic richness.
\newblock {\em Theor. Comput. Sci.}, 410(30-32):3018--3026, 2009.
\newblock \href {https://doi.org/10.1016/j.tcs.2009.03.008}
  {\path{doi:10.1016/j.tcs.2009.03.008}}.

\bibitem{2018RubinchikS_EertrEfficDataStrucFor_EJC}
Mikhail Rubinchik and Arseny~M. Shur.
\newblock {EERTREE:} an efficient data structure for processing palindromes in
  strings.
\newblock {\em Eur. J. Comb.}, 68:249--265, 2018.
\newblock \href {https://doi.org/10.1016/j.ejc.2017.07.021}
  {\path{doi:10.1016/j.ejc.2017.07.021}}.

\end{thebibliography}

\end{document}